\documentclass[11pt]{article}

\usepackage{times}
\usepackage{paralist}
\usepackage{amsfonts}
\usepackage{amsmath}
\usepackage{amsthm}
\usepackage{latexsym}
\usepackage{url,hyperref}
\usepackage{picins,boxedminipage}
\usepackage{microtype}

\usepackage{amssymb}
\usepackage{accents}
\usepackage{color}

\usepackage{float}
\usepackage{graphicx}

\graphicspath{{graphics/}}

\usepackage[paper=letterpaper,margin=1in]{geometry}

\newtheorem{theorem}{Theorem}
\newtheorem*{theorem*}{Theorem}

\newtheorem*{conjecture*}{Conjecture}

\newtheorem*{remark*}{Remark}
\newtheorem{lemma}[theorem]{Lemma}
\newtheorem*{lemma*}{Lemma}

\newtheorem*{proposition*}{Proposition}
\newtheorem{corollary}[theorem]{Corollary}
\newtheorem*{corollary*}{Corollary}
\newtheorem{property}[theorem]{Property}
\newtheorem*{property*}{Property}

\def\b1{{\bf 1}}

\title{Matroidal Degree-Bounded Minimum Spanning Trees}

\author{Rico Zenklusen\thanks{Dept.~of Mathematics, MIT,
Cambridge. E-mail: {\tt ricoz@math.mit.edu}.
Supported by Swiss National Science Foundation grant
PBEZP2-129524, by NSF grants CCF-1115849 and CCF-0829878, and by
ONR grants N00014-11-1-0053 and N00014-09-1-0326.
}}

\begin{document}

\maketitle

\begin{abstract}

We consider the minimum spanning tree (MST) problem under the
restriction that for every vertex $v$, the edges of the tree
that are adjacent to $v$ satisfy a given family of constraints.
A famous example thereof is 
the classical degree-bounded MST problem, where for every vertex
$v$, a simple upper bound on the degree is imposed.
Iterative rounding/relaxation algorithms became the tool of choice for
degree-constrained network design problems.
A cornerstone for this development
was the work of Singh and Lau~\cite{singh_2007_approximating},
who showed that for the degree-bounded MST problem,
one can find a spanning tree violating each degree bound
by at most one unit and with cost at most the cost of an optimal solution
that respects the degree bounds.

However, current iterative rounding approaches face several limits
when dealing with more general degree constraints.
In particular, when several constraints are imposed on the edges
adjacent to a vertex $v$, as for example when a partition of
the edges adjacent to $v$ is given and only a fixed number of
elements can be chosen out of each set of the partition,
current approaches might violate each of the constraints by
a constant, instead of violating the whole family of constraints
by at most a constant number of edges.
Furthermore, it is also not clear how previous iterative rounding
approaches can be used for degree constraints where some edges are in
a super-constant number of constraints.

We extend iterative rounding/relaxation approaches both
on a conceptual level as well as aspects involving their analysis
to address these limitations.
Based on these extensions, we 
present an algorithm for the degree-constrained MST problem
where for every vertex $v$, the edges adjacent to $v$ have
to be independent in a given matroid. The algorithm returns
a spanning tree of cost at most OPT such that for every vertex $v$,
it suffices to remove at most $8$ edges from the spanning tree
to satisfy the matroidal degree constraint at $v$.
\end{abstract}

\section{Introduction}

Recently, much effort has been put on designing approximation algorithms
for degree-constrained network design problems.
This development was motivated by various applications as for example
VLSI design, vehicle routing, and applications in communication
networks~\cite{fuerer_1994_approximating,
bauer_1995_degreeconstrained, ravi_2001_approximation}.
One of the most prominent and elementary problems here, which 
attracted lots of attention in recent years, are 
degree-constrained (MST) problems.

In the most classical setting, known as the
\emph{degree-bounded MST problem}, the problem is to find a 
spanning tree $T\subseteq E$ of minimum cost in a graph $G=(V,E)$
under the restriction that the degree of each vertex $v$ with
respect to $T$ is at most some given value $B_v$.
Since checking feasibility of a degree-bounded MST problem is already
NP-hard, interest arose in finding low-cost spanning trees that violate
the given degree constraints slightly.
A long chain of
papers (see~\cite{fuerer_1994_approximating, koenemann_2002_matter,
koenemann_2003_primaldual, chaudhuri_2009_pushrelabel,
chaudhuri_2009_what} and references
therein)
led to algorithms with various trade-offs between cost of the
spanning tree and violation of the degree bounds.
In recent years, important progress was achieved for the degree-bounded
MST problem, which also led to a variety of new techniques.
Goemans~\cite{goemans_2006_minimum} showed how to find a spanning tree
violating each degree constraint by at most two units, and whose cost is
bounded by the cost OPT of an optimum spanning tree that satisfies the
degree constraints.
Enhancing the iterative rounding framework introduced by
Jain~\cite{jain_2001_factor} with a relaxation step,
Singh and Lau~\cite{singh_2007_approximating} obtained a stronger
version of the above result, which is essentially best possible,
where degree constraints are only violated by at most one unit.
They work with an LP relaxation of the problem, and
iteratively drop degree constraints from the LP that cannot 
be violated by more than one unit in later iterations.
The adapted LP is then solved again to obtain a possibly
sparser basic solution that allows for further degree
relaxations.
Edges not used in the current optimal solution to the LP are removed
from the graph, and edges that have a weight of one are fixed,
while updating degree bounds accordingly.
A degree bound at a vertex $v$ is removed whenever 
it is at most one unit lower
than the current number of edges adjacent to $v$.

We are interested in obtaining results of similar strength
for more general degree bounds.
Consider for example the following type of degree constraints:
for every vertex $v$, a partition $E^v_1,\dots, E^v_{n_v}$
of the set $\delta(v)$ of edges adjacent to $v$ is given,
and within each set $E^v_i$ of the partition, only a given number of edges
can be chosen.
The algorithm of Singh and Lau~\cite{singh_2007_approximating}
as well as the one of Goemans~\cite{goemans_2006_minimum}
can easily be adapted
to this setting. (In particular, the algorithm of Singh and Lau
was even presented in this precise setting.)
However, with both of these approaches, the constraint imposed by
each set $E^v_i$ can be violated by a constant.
We are interested in having at most a constant violation
over all degree constraints at $v$,
i.e., for every vertex $v\in V$, at most a constant number of edges
have to be removed from the spanning tree to satisfy all
constraints at $v$.
Another more general example that will be useful to
illustrate limits of current methods
is obtained when imposing constraints for
each vertex $v$ on a laminar family on the edges adjacent to $v$, instead
of only considering a partition.

Adapting Goemans' algorithm to these stricter bounds on the degree violation
seems to be difficult, since a crucial step of this algorithm is to
cover the support $E^*$ of a basic solution to the natural LP relaxation
by a constant number of spanning trees (for the degree-bounded MST
problem, Goemans showed~\cite{goemans_2006_minimum} that two spanning
trees suffice).
This result allows for orienting the edges in $E^*$ such that every vertex
has at most a constant number of incoming arcs, at most one in each spanning
tree. Dropping for every vertex all incoming arcs from its degree constraint
then leads to a matroid intersection problem, whose solution violates
each degree constraint by at most a constant.
To be able to decompose $E^*$ into a constant number of spanning trees, one
needs to show that for any subset of the vertices $S\subseteq V$, only
a linear number (in $|S|$) of edges have both endpoints in $S$.
In the classical degree-bounded MST problem, this 
sparseness property follows from the fact that when considering only
edges with both endpoints in $S$, there are at most a linear number
(in $|S|$) of linearly independent and tight spanning tree constraints due
to combinatorial uncrossing, and only a linear number of
degree constraints within $S$.
However, in more general settings as highlighted above, the number
of degree constraints within $S$ can be super-linear.

Iterative relaxation looks more promising for a possible
extension to generalized degree bounds. 
However, current iterative rounding approaches face several
limits when trying to adapt them.
In particular, when dealing with the partition bounds as explained
above, a simple adaptation of the relaxation rule, where for a
vertex $v$ all constraints at $v$ would be dropped as soon as
it is safe to do so due to a small support $E^*$, risks
to get stuck because there might be no 
vertex whose degree constraint can be relaxed.
Furthermore, previous approaches (as used in~\cite{singh_2007_approximating,
bansal_2009_additive}) to show that the support is sparse
fail in our setting because of a possible super-linearity of the total number
of degree constraints.
Additionally, previous iterative relaxation approaches
crucially rely on the property that any edge is in at
most a constant number of degree constraints to obtain
violations that are bounded by a constant.
However, this does not hold when dealing for example
with degree constraints given by a laminar family.

In this paper we show how to extend iterative relaxation approaches,
both from a conceptual point of view as well as aspects involving
their analysis, to tackle a wide class of MST problems with generalized
degree bounds, namely when the degree bounds for every vertex are given
by a matroid. In particular, this includes the partition bounds
and the more general laminar bounds mentioned above.

\paragraph{\textbf{Our results.}}

We present an iterative rounding/relaxation algorithm for finding a
\emph{matroidal degree-bounded MST}. The degree bounds are given as follows:
for every vertex $v$, a matroid $M_v=(\delta(v),\mathcal{I}_v)$ over
the ground set $\delta(v)$ is given with independent sets
denoted by $\mathcal{I}_v\subseteq 2^{\delta(v)}$.
The problem (without relaxed degree constraints) is to
find a spanning tree $T$ in $G$ satisfying
$T\cap \delta(v)\in \mathcal{I}_v$ $\forall v\in V$,
and minimizing a linear cost function $c:E \rightarrow \mathbb{R}_+$.
We say that a given spanning tree $T$ violates a degree constraint
$M_v$ by at most $k\in\mathbb{N}$ units, if it suffices to remove at
most $k$ edges $R\subseteq \delta(v)\cap T$ from $T$ to satisfy
the constraint $M_v$, i.e.,
$(T\setminus R)\cap \delta(v)\in \mathcal{I}_v$.
Hence, the partition and laminar bounds mentioned above correspond
to the case where all matroids $M_v$ are partition or laminar matroids,
respectively. We show the following.
\begin{theorem}
There is an efficient algorithm for the matroidal degree-bounded MST
problem that returns a spanning tree of cost at most the cost
of an optimal solution, and violates each degree bound by at most
$8$ units.
\end{theorem}
To overcome problems faced by previous iterative relaxation approaches,
we enhance the iterative relaxation step, and exploit polyhedral
structures to prove stronger sparseness results.
The polytope used as a relaxation of the matroidal degree-bounded
MST asks to find a point $x\in \mathbb{R}^E$ in the spanning tree
polytope such that for every vertex $v\in V$, the restriction
of $x$ to $\delta(v)$ is in the matroid polytope $P_{M_v}$ of $M_v$.

To be able to always find possible relaxation steps, our
iterative rounding procedure tries to achieve a somewhat weaker goal
than previous approaches. The algorithm of Singh and
Lau~\cite{singh_2007_approximating} relaxes degree constraints with
the goal to approach the spanning tree polytope, which is integral.
In our approach, the goal we pursue is to remove every edge $\{u,v\}$
from at least one of the two degree constraints at $u$ or $v$.
As soon as no edge
is part of both degree constraints at its endpoints, the problem
is a matroid intersection problem, since all degree constraints together
can be described by a single matroid over the support of the current
LP solution.
Thus, once we are in this situation, the current LP will be integral
and no further rounding steps are needed.
Hence, in our relaxation step,
we try to find a vertex $v$ such that we
can remove all edges adjacent to $v$ that are still in both degree
constraints from the degree constraint at $v$. 
Edges adjacent to $v$ that are only contained in the degree
constraint at $v$ will not be removed from the constraint $M_v$.
Our approach has thus some similarities with Goemans' method,
but instead of removing right at the start every edge from one
degree constraint, we do this iteratively and hereby profit from
additional sparseness that is obtained by solving the
LP relaxation after each degree adaptation step.
As we will see in Section~\ref{sec:algo}, the way
how we remove edges from a constraint is strictly speaking
not a relaxation, and we therefore prefer to use the term
\emph{degree adaptation} instead of degree relaxation.
The above degree adaptation step alone shows not to be sufficient
for our approach, since one might still end up in a situation
were no further degree adaptation can be performed because
the graph is too dense.
To obtain greater sparsity, we use a second type of degree adaptation,
where for some vertex $v$ we remove (almost) the full degree
constraint at $v$ if this cannot lead to a large violation of
the degree constraint at $v$.

The main step in the analysis is to prove that it is always possible
to apply at least one of two suggested degree adaptations. 
A first step in this proof is to show that the support of a basic
solution to the LP relaxation is sparse.
We obtain sparsity by showing that if there are $k\in \mathbb{N}$
linearly independent and tight constraints (with respect to the
current LP solution $x$) of the polytope $P_{M_v}$, then
$x(\delta(v))\geq k$. Since summing $x(\delta(v))$ over all vertices
is equal to $2(|V|-1)$, because $x(E)=|V|-1$ as $x$ is in the
spanning tree polytope, there are at most $2(|V|-1)$ linearly
independent and tight degree constraints.

The crucial part in the analysis is to show that
vertices to which no further degree adaptation can be performed
do not have very low degrees in average, implying that some
of the other vertices are likely to
have low degrees and therefore admit a degree constraint adaptation.
To prove this property, we exploit the interplay between degree bounds
and spanning tree constraints to show that any degree two node
can either be treated separately and allows for reducing the problem, or
implies a reduction in the maximum number of linearly independent and
tight spanning tree constraints.

\paragraph{\textbf{Related work.}}

The study of spanning trees with degree constraints can be traced
back to F\"urer and Raghavachari~\cite{fuerer_1994_approximating},
who presented an approximation algorithm for
the degree-bounded Steiner Tree problem which
violates each degree bound by at most one, but does
not consider costs.
This result generated much interest in the study of degree-bounded
network design problems, leading to numerous results and new
techniques in recent years for a variety of problems,
including degree-bounded arborescence problems,
degree-bounded $k$-edge-connected subgraphs,
degree-bounded submodular flows,
degree-bounded bases in
matroids (see~\cite{raghavachari_1996_algorithms,
ravi_2001_approximation,
klein_2004_approximation,
lau_2007_survivable,
lau_2008_additive,
kiraly_2008_degree,
bansal_2009_additive,
chekuri_2010_dependent,
bansal_2010_generalizations} and references therein).

Spanning tree problems with a somewhat different notion
of generalized degree bounds have been considered
in~\cite{bansal_2009_additive} and~\cite{bansal_2010_generalizations}.
In these papers, the term ``generalized degree bounds'' is used
as follows: given is a family of sets $E_1,\dots, E_k\subseteq E$,
and the number of edges that can be chosen out of each set
$E_i$ is bounded by some given value $B_i\in\mathbb{N}$.
In~\cite{bansal_2009_additive}, using an iterative relaxation algorithm,
whose analysis is based on a fractional token counting
argument, the authors show how to efficiently obtain a spanning tree
of cost at most OPT and violating each degree bound by at
most $\max_{e\in E}|\{i\in [k] \mid e\in E_i\}|$, the maximum
coverage of any edge by the sets $E_i$.
In~\cite{bansal_2010_generalizations}, a new iterative rounding approach
was presented for the problem when the sets $E_1,\dots, E_k$ correspond
to the edges $E_i=\delta(C_i)$ of a family of cuts $C_i\subseteq V$
for $i\in [k]$ that is laminar.
Contrary to previous settings where iterative rounding approaches
were applied, here, it is possible that an edge lies in a
super-constant number of degree constraints.
At each iteration, the algorithm reduces the number of degree
constraints by a constant factor, replacing some constraints with
new ones if necessary. This is done in such a way that degree
constraints are violated by at most a constant in every iteration,
leading to a spanning tree of cost at most OPT, that violates
each degree constraint by at most $O(\log(|V|))$.

\paragraph{\textbf{Organization.}}
In Section~\ref{sec:algo} we present
our algorithm for the matroidal degree-bounded MST problem.
The analysis of the algorithm is presented in Section~\ref{sec:analysis}.

\section{The algorithm}\label{sec:algo}

Since during the execution of our algorithm the underlying graph
will be modified, we denote by $H=(W,F)$ the current state of
the graph, whereas $G=(V,E)$ always denotes the original graph.
For brevity, terminology and notation is with respect to the current
graph $H$ when not specified further.
To distinguish between initial degree constraints and current degree
constraints, we denote by $N_w$ the current constraints
for $w \in W$---which will as well be of matroidal type---whereas
$M_v$ denotes the initial degree constraints 
at $v\in V$. The vertices of $H$ are called \emph{nodes}
since they might contain several vertices of $G$ due to
edge contractions.

The algorithm starts with $H=G$ and $N_v=M_v$ for $v\in V$, and the
LP relaxation we use is the following,
\begin{equation*}
(LP1)\qquad
\boxed{
\begin{array}{crcll}
\min & c^T x &&& \\
& x & \in & P_{st} & \\
& x\big |_{\delta(w)} & \in & P_{N_w} & \forall w\in W\\
\end{array}
}
\vspace{0.2em}
\end{equation*}
where $P_{st}$
denotes the spanning tree polytope of $H$, 
$P_{N_w}$ denotes the matroid
polytope that corresponds to $N_w$, and
$x\big |_{\delta(w)}$ denotes the vector obtained from $x\in \mathbb{R}^E$
by considering only the components that correspond to $\delta(w)$.

There is a set of nodes $Q=Q(H,x)\subseteq W$ that has a special role
in our algorithm due to its relation with tight spanning tree
constraints. 
The node set $Q$ is defined and used after having contracted edges 
of weight one. Hence, assume that $H$ does not contain any edge
$f\in F$ with $x(f)=1$.
Then $Q$ is defined as follows:
we start with $Q=\emptyset$ and as long as there is a node $w\in W$
such that $x(\delta(w)\cap F[W\setminus Q])=1$, where $F[W\setminus Q]$
is the set of all edges with both endpoints in $W\setminus Q$,
we add $w$ to $Q$. One can easily observe that $Q$ does not dependent
on the order in which nodes are added to it\footnote{The fact that $H$
does not contain $1$-edges is needed here to make sure that the order
is unimportant in the definition of $Q$. With $1$-edges it might
be that during the iterative construction of $Q$, one ends up
with two nodes connected by a single edge of weight one,
in which case any one of the two remaining nodes can
be included in $Q$, but not both. This is actually the only bad constellation
that leads to a dependency on the order in the definition
of $Q$.
}.
As we will see later, edges adjacent to these nodes can often be
ignored from degree constraints due to strong restrictions that
are imposed by the spanning tree constraints.

The box on top of the page gives a description of our algorithm, omitting 
details of how to deal with the matroidal degree bounds when
removing or contracting edges.
We discuss these missing points in the following.
\begin{figure}
\begin{center}
\begin{boxedminipage}{0.92\linewidth}
\textbf{Algorithm for Matroidal Degree-Bounded Minimum Spanning Trees}
\begin{enumerate}
\item Initialization: $H=(W,F) \leftarrow G=(V,E)$, $N_v\leftarrow M_v$
for $v\in V$.
\item While $|W|>1$ do 
\begin{enumerate}[a)]
\item\label{enum:solveLP} Determine basic optimal solution $x$ to $(LP1)$.
Delete all edges $f\in F$ with $x(f)=0$.
\item\label{enum:contract} Contract all edges $f\in F$ with $x(f)=1$.
\item\label{enum:fixConstraints} Fix a maximal family
of linearly independent and tight spanning tree constraints.
\item\label{enum:relaxA}
\emph{Type A degree adaptation}: for each node $w\in W$ such that
the set of all edges $U\subseteq \delta(w)$ that are still in
both degree constraints is non-empty and satisfies
$|U|-x(U)\leq 4,$
remove $U$ from the degree constraint $N_w$.
\item\label{enum:relaxB}
\emph{Type B degree adaptation}: for each node $w\in W$ such that
the set of all edges $U\subseteq \delta(w)$ contained in
the degree constraint $N_w$ but not adjacent to a node in $Q$
is non-empty and satisfies
$|U|-x(U)\leq 4,$
remove $U$ from the degree constraint $N_w$.
\end{enumerate}
\item Return all contracted edges.
\end{enumerate}
\end{boxedminipage}
\end{center}
\end{figure}

Notice, that a basic solution to $(LP1)$ can be determined
in polynomial time by the ellipsoid method, even if the
involved matroids are only accessible trough an independence oracle.
Depending on the matroidal degree bounds involved, $(LP1)$ can 
be solved more efficiently by using a polynomially-sized extended
formulation.

A tight spanning tree constraint, as considered in
step~\eqref{enum:fixConstraints}, corresponds to a set
$L\subseteq W, L\neq \emptyset$
such that $x(F[L])=|L|-1$.
\emph{Fixing} a tight spanning tree constraint means that
this constraint has to be fulfilled with equality in all
linear programs of type $(LP1)$ solved in future iterations.
It is well-known that if $supp(x)=F$, then any maximal family
of linearly independent and tight spanning tree constraints
with respect to $x$ defines the minimal face of the spanning
tree polytope on which $x$
lies~(see e.g.~\cite{goemans_2006_minimum}).
Hence, due to step~\eqref{enum:fixConstraints}, we have that
if the LP solution at some iteration of the algorithm is on
a given face of the spanning tree polytope, then all future
solutions to $(LP1)$ will be as well on this face.

Fixing tight spanning tree constraints shows to be useful since
they 
often imply strong conditions on the edges, which
can be exploited when having to make sure that degree
constraints are not violated too much. In particular, consider
a node $w\in Q$ which, in the iterative construction of $Q$,
could have been added as the first node, i.e., $x(\delta(w))=1$.
When fixing tight spanning tree constraints, one can observe that
any spanning tree satisfying those tight constraints with equality
contains precisely one edge adjacent to $w$.
Furthermore, the fixing of tight spanning tree constraints guarantees
that a node $w\in Q$ will stay in $Q$ in later iterations until
an edge adjacent to $w$ is contracted.
Hence, all edges being in some iteration adjacent to a node $w\in Q$,
will be adjacent to a node in $Q$ in all later iterations until
they are either deleted or contracted.
This property is important in our approach since a type B
degree adaptation ignores edges adjacent to $Q$, and we want
to make sure that an edge which is once ignored will never
be considered during a later type B degree adaptation.

\vspace{0.5em}
\noindent\textbf{Contracting and removing edges.}
To fill in the remaining details of our algorithm,
it remains to discuss how edges are contracted
and removed.
Throughout the algorithm, any degree constraint $N_w$ 
of a node $w$ containing the vertices $v_1,\dots, v_k\in V$
can always be written as a disjoint union of matroidal constraints
$N_{v_1}, \dots, N_{v_k}$, where $N_{v_i}$ corresponds to
the ``remaining'' degree bound at $v_i$ and is a matroid
over the edges $\delta(w)$ that are adjacent to $v_i$.
Whenever an edge $f=\{w_1,w_2\}$ of weight one is contracted in
step~\eqref{enum:contract} of the algorithm to obtain a new node $w$,
the new degree constraint $N_w$ at $w$ is obtained by taking a disjoint
union of the matroids $N_{w_1}/ f$ and $N_{w_2}/ f$, where
$N_{w_1}/ f$ and $N_{w_2}/ f$ correspond to the matroids obtained
from $N_{w_1}$ and $N_{w_2}$, respectively, by contracting $f$.
This operation
simply translates the degree constraints on $w_1$ and $w_2$ to
the merged node $w$.
The property that a degree bound on $w$ is a disjoint
union of degree bounds of the vertices represented by $w$,
is clearly maintained by this contraction.

As highlighted in the box, we adapt constraints by
\emph{removing} for some node $w\in W$
a set of edges $U\subseteq\delta(w)$ from the constraint $N_w$.
When removing $U$ from $N_w$, we construct a new degree constraint
given by a matroid $N_w'$ over the elements $\delta(w)$ such that
the following properties hold.
\begin{property}\label{prop:edgeRemoval}
\hspace{0mm}\\[-1em]
\begin{compactenum}[i)]
\item\label{item:disjoint} $N_w'$ is
a disjoint union of
matroidal constraints $N'_{v_1}, \dots, N'_{v_k}$
corresponding to vertices contained in $w$,
\item\label{item:freeElems} edges of $U$ are \emph{free elements}
of $N_w'$, i.e., if
$I$ is independent in $N_w'$ then $I\cup U$ is 
independent in $N_w'$,
\item\label{item:indepTrans} any independent set of $N'_{v_i}$
can be transformed into one of $N_{v_i}$ by removing at
most $\lceil|U|-x(U)\rceil$ edges, 
\item\label{item:stillFeasible} the previous LP solution
$x$ is still feasible with respect to $N_w'$,
i.e., $x\big |_{\delta(w)}\in P_{N_w'}$.
\end{compactenum}
\end{property}

Any removal operation satisfying the above
properties can be used in our algorithm.
Before presenting such a removal operation,
we first mention a few important points.
To avoid confusion, we want to highlight that removing $U$
from $N_w$ does not simply correspond to deleting the elements
$U$ from the matroid $N_w$.
For any edge $f\in \delta(w)$ that is free in $N_w$, we say that $f$
is \emph{not contained} in the degree constraint $N_w$, and it is
\emph{contained} otherwise. When all edges adjacent to a given
node $w$ are not contained in its degree constraint, which corresponds
to $N_w$ being a free matroid, we say that the node $w$ has no
degree constraint.

We now discuss how to remove a set of edges $U\subseteq \delta(w)$ from
$N_w$ to obtain an adapted degree bound $N'_w$ satisfying
Property~\ref{prop:edgeRemoval}.
Let $v_1,\dots, v_p\in V$ be all vertices contained in the node $w$,
and we consider the decomposition of $N_w$
into a disjoint union of matroids $N_{v_1}, \dots, N_{v_k}$, 
where $N_{v_i}$ for $i\in [k]$ corresponds to the
``remaining'' degree bound at $v_i$.
To remove $U$ from $N_w$, we adapt each matroid $N_{v_i}$ as follows
to obtain a new matroid $N'_{v_i}$. Let $S_i$ be the ground
set of $N_{v_i}$, i.e., all edges in $\delta(w)$ being adjacent
to $v_i$.
Let $M_1=(S_i,\mathcal{I}_1)$ be the matroid with independent sets
\begin{equation*}
\mathcal{I}_1=\{I\subseteq S_i\cap U \mid
|I|\leq |S_i\cap U|-\lfloor x(S_i\cap U) \rfloor\}.
\end{equation*}
Hence, $M_1$ is a special case of a partition matroid.
Let $M_2=M_1 \vee N_{v_i}$ be the union of the matroids
$M_1$ and $N_{v_i}$,
and let $M_3=M_2/(S_i\cap U)$ be the matroid obtained from
$M_2$ by contracting
$S_i\cap U$. The degree bound $N'_{v_i}$ is obtained 
by a disjoint union of $M_3$ and a free matroid over the elements
in $S_i\cap U$. The new degree constraint $N'_w$, that results
by \emph{removing} $U$ from $N_w$, is given by
the disjoint union of the matroids $N'_{v_1}, \dots, N'_{v_k}$.

\begin{lemma}\label{lem:removalProcOK}
The above procedure to remove elements from a degree constraint
satisfies Property~\ref{prop:edgeRemoval}.
\end{lemma}
\begin{proof}
By construction, when removing a set $U\subseteq \delta(w)$
from a degree bound $N_w$, which can be written as a disjoint
unions of $N_{v_1}, \dots, N_{v_k}$, a matroidal bound
$N'_w$ is determined which is a disjoint union of
$N_{v_1}',\dots,N_{v_k}'$.
Hence point~\eqref{item:disjoint} of Property~\ref{prop:edgeRemoval}
holds.

Let $S_i$ be the ground set of the matroids $N'_{v_i}, N_{v_i}$
for $i\in [k]$.
Since $N'_w$ is a disjoint union of $N'_{v_1},\dots, N'_{v_k}$,
it suffices for point~\eqref{item:freeElems} to prove
that if $I'$ is independent in $N'_{v_i}$ then $I'\cup (S_i\cap U)$
is independent in $N'_{v_i}$.
This follows since $N'_{v_i}$ was obtained by a disjoint union of
the matroid $M_3$, as defined above, and
a free matroid over $S_i\cap U$.

For point~\eqref{item:indepTrans}, consider an independent set
$I'$ in $N_{v_i}'$.
Since all edges in $U\cap S_i$ are free in $N_{v_i}'$, we can assume
$(U\cap S_i)\subseteq I'$.
Consider how the matroid $N_{v_i}'$ was constructed by the
use of the matroids $M_1, M_2, M_3$. 
We start by observing that $U\cap S_i$ is an independent set
in $M_2=M_1\vee N_{v_i}$.
Let $r_i$ be the rank function of $N_{v_i}$, and
$r_2$ be the rank function of $M_2$. Since $x\in P_{N_{v_i}}$,
we have that $r_i(S_i\cap U)\geq x(S_i\cap U)$.
Furthermore, since $M_2=M_1\vee N_{v_i}$ and any
$|S_i\cap U|-\lfloor x(S_i\cap U)\rfloor$ elements of $S_i\cap U$
are independent in $M_1$, we have
\begin{equation*}
r_2(S_i\cap U)= \min\{|S_i\cap U|, r_i(S_i\cap U)+
|S_i\cap U|-\lfloor x(S_i\cap U)\rfloor\}=|S_i\cap U|,
\end{equation*}
showing independence of $S_i\cap U$ in $M_2$.
Because $N_{v_i}'$ was obtained by a disjoint
union of the matroid $M_3$ and
a free matroid over the elements $S_i\cap U$, we can write
$I'=I_3\cup (S_i\cap U)$ with $I_3$ independent in $M_3$.
Furthermore, as $M_3= M_2/ (S_i\cap U)$ and $S_i\cap U$ is
independent in $M_2$, the set $I'$ is independent
in $M_2$.
As $M_2=M_1\vee N_{v_i}$, we have
$I'=I_1 \cup I$, with $I_1$ independent in $M_1$ and
$I$ independent in $N_{v_i}$.
Since $M_1$ is a matroid
of rank $|S_i\cap U|-\lfloor x(S_i\cap U)\rfloor$, we have
that $I$ is obtained from $I'$ by removing at most
$|I_1|\leq |S_i \cap U|-\lfloor x(S_i\cap U) \rfloor
\leq |U|-\lfloor x(U) \rfloor$ elements as desired.

Let $x_i = x \big|_{S_i}$ for $i\in [k]$.
To show point~\eqref{item:stillFeasible}, it suffices to
prove that $x_i\in P_{N'_{v_i}}$ $\forall i\in [k]$,
since $N_w'$ is a disjoint union of $N'_{v_1},\dots,N'_{v_k}$.
Let $z_i\in [0,1]^{S_i}$ be given by
\begin{equation*}
z_i(f) = \begin{cases}
1      & \text{if } f \in S_i\cap U,\\
x_i(f) & \text{if } f\in S_i\setminus U.
\end{cases}
\end{equation*}
Observe that $z_i-x_i \in P_{M_1}$ because
the support of $z_i-x_i$ is a subset of $S_i\cap U$,
$\lVert z_i-x_i \rVert_1 = |S_i\cap U|- x(S_i\cap U)$
and any $|S_i\cap U|-\lfloor x(S_i\cap U) \rfloor$ elements
of $S_i\cap U$ are independent in $M_1$.
Hence $z_i\in P_{M_2}$, since $M_2=M_1\vee N_{v_i}$,
$x_i \in P_{N_{v_i}}$ and $z_i-x_i\in P_{M_1}$.
As $M_3=M_2/(S_i\cap U)$, we have that the restriction
of $z_i$ on $S_i\setminus U$, which is equal to
$x_i\big |_{S_i\setminus U}$, is in $P_{M_3}$.
Since $N_{v_i}'$ is the union of $M_3$ and a
free matroid over $S_i\cap U$, this finally implies
that $x_i\in P_{N'_{v_i}}$.

\end{proof}

\section{Analysis of the algorithm}\label{sec:analysis}

\begin{lemma}\label{lem:atMost2Adaptations}
During the execution of the algorithm, for every vertex $v\in V$,
at most one constraint adaptation of type A and one of type B is
performed that removes edges of $\delta(v)$ from degree
constraints containing $v$.
\end{lemma}
\begin{proof}
When a type A degree adaptation is applied to a node
$w\in W$ that contains $v$,
no further type A degree adaptation can remove any edges
in $\delta(v)\cap F$ from the constraint
containing $v$, since those edges are not anymore contained
in both degree constraints at their endpoints.

Similarly, when a type B degree adaptation is applied
to a node $w$ that contains $v$, all edges in
$\delta(w)\cap F[W\setminus Q]$ are removed from
the degree constraint at $w$ and thus cannot be
removed again at a later type B degree adaptation.
Hence, the only possibility to remove further edges
adjacent to $v$ in a later type B degree adaptation is that some
edge $f\in F$ which was---at some iteration of the algorithm---not
considered for a possible removal by a type B adaptation
because of being adjacent to a node in $Q$, can be removed
by a type B adaptation at a later stage.
However as already discussed,
since we fix all tight spanning tree constraints, 
an edge that is adjacent to a node $w\in Q$ in some
iteration, will remain so until it is either
deleted or contracted in step~\eqref{enum:solveLP}
or~\eqref{enum:contract} of the algorithm.
Hence, this ``bad constellation'' can never occur.
\end{proof}

We exploit that our removal operation satisfies
point~\eqref{item:indepTrans} of Property~\ref{prop:edgeRemoval} to bound
the maximum possible degree violation.
In particular, for each vertex $v\in V$,
every time edges $U$ with $U\cap \delta(v)\neq \emptyset$
are removed from the current degree constraint $N_w$ at the
node $w$ that contains $v$, the degree constraint at $v$
can be violated at most by an additional $\lceil |U| - x(U) \rceil$ units.
Since we only perform degree adaptations for sets $U$ with
$|U|-x(U)\leq 4$, and Lemma~\ref{lem:atMost2Adaptations} guarantees
that at most two adaptations are performed that involve the degree constraint
at $v$, we obtain the following result.
\begin{corollary}
If the algorithm terminates, then the returned tree violates each
degree constraint by at most $8$ units.
\end{corollary}

A main step for proving that we can always apply one of the two
suggested degree adaptations, is to prove that a basic solution to
$(LP1)$ is sufficiently sparse.
A first important building block for proving sparsity is the
following result.

\begin{lemma}\label{lem:indepDegConstraints}
Let $x$ be any solution to $(LP1)$ whose support equals $F$.
Then for every node $w\in W$, the maximum number of linearly independent
constraints of the matroid polytope $P_{N_w}$ that are tight with
respect to $x$, is bounded by $x(\delta(w))$.
\end{lemma}
\begin{proof}
Let $\mathcal{C}\subseteq 2^{\delta(w)}$ be a family with a maximum number
of sets that correspond to linearly independent constraints of the
matroid polytope $P_{N_w}$ that are tight with respect to $x$.
By standard uncrossing arguments, $\mathcal{C}$ can be chosen to be
a chain, i.e.~$\mathcal{C}=\{C_1,\dots, C_p\}$ with
$C_1\subsetneq C_2 \subsetneq \dots \subsetneq C_p$
(see~\cite{jain_2001_factor} for more details).
We have to show that $p\leq x(\delta(w))$.
Let $r$ be the rank function of $N_w$.
Define $C_0=\emptyset$ and for 
$i\in [p]$ let $R_i=C_{i}\setminus C_{i-1}$.
Since $\mathcal{C}$ is a family of tight constraints, we have
\begin{equation}\label{eq:tightChain}
x(R_i)=r(C_{i})-r(C_{i-1}) \quad \forall i\in [p].
\end{equation}
Because $R_i\subseteq supp(x)$, the left-hand side of~\eqref{eq:tightChain}
is strictly larger than zero. Furthermore, the right-hand side is integral and
must therefore be at least one. Hence $x(R_i)\geq 1$ for $i\in [p]$,
which implies $x(\delta(v))\geq \sum_{i=1}^p x(R_i)\geq p$.
\end{proof}

Notice, that the above lemma implies that a basic solution $x$ 
to $(LP1)$ has a support of size at most $3(|W|-1)$,
because of the following.
We can assume that all edges that are not in the support of
$x$ are deleted from the graph.
Due to Lemma~\ref{lem:indepDegConstraints}, at most
$\sum_{w\in W}x(\delta(w))$
linearly independent constraints of the polytopes $\{P_{N_w}\mid w\in W\}$
can be tight with respect to $x$, and since $x$ is in the spanning
tree polytope of $H$, this bound equals $\sum_{w\in W}x(\delta(w))=2(|W|-1)$.
Furthermore, at most $|W|-1$ linearly independent constraints
of $P_{st}$ are tight with respect to $x$ due to uncrossing.
This shows in particular that in the first iteration of the algorithm,
we can find a node $w\in W$ to which a type A degree constraint adaptation
can be applied, because
\begin{equation*}
\sum_{w\in W}\left(|\delta(w)|-x(\delta(w)\right)= 2|F|-2(|W|-1)\leq 4(|W|-1),
\end{equation*}
and hence there must be a node $w\in W$ with $|\delta(w)|-x(\delta(w))\leq 4$.

However, in later iterations, the above reasoning alone is not anymore
sufficient because many vertices do not have degree constraints anymore.
Still, by assuming that no type B constraint adaptation is possible,
and using several ideas to obtain stronger sparsity,
we show that the above approach of finding a good vertex for a type A degree
adaptation by an averaging argument can be extended to a general
iteration.

For the rest of this section, we consider an iteration of the algorithm
at step $\eqref{enum:relaxA}$ with a current basic
solution $x$ to $(LP1)$,
and assume that $|W|>1$, and that no type B degree adaptation can
be applied%
\footnote{Notice that the assumption $|W|>1$ is not redundant.
Whereas we know that at the beginning of the iteration $|W|>1$ did
hold, this could have changed after contracting edges in
step~\eqref{enum:contract}.}.
We then show that there is a type A constraint adaptation that 
can be performed under these assumptions.
This implies that our algorithm never gets
stuck, and hence proves its correctness.

Since we often deal with the \emph{spare} $1-x(f)$ of an edge
$f\in F$, we use the notation $z=1-x$.
Furthermore, we partition $F$ into the sets $F_2, F_1$ and $F_0$
of edges that are contained in $2,1$ and $0$ degree constraints,
respectively. Hence, at the first iteration we have $F_2=F$.
Our goal is to show that
$\sum_{w\in W}z(\delta(w)\cap F_2)= 2z(F_2)\leq 4|Y|$, where
$Y\subseteq W$ is the set of all nodes $w$ with
$\delta(w)\cap F_2\neq \emptyset$. By an averaging argument this 
then implies that there is at least one node $w\in Y$ to which a type
A constraint adaptation can be applied.
Notice that the set $F_2$ cannot be empty
(and hence also $Y\neq \emptyset$): if $F_2=\emptyset$,
then the current $LP1$ corresponds to a matroid intersection problem
since every edge is contained in at most one degree constraints,
and hence all degree constraints form together a single 
matroid over $F$; in this case $LP1$ is integral and a full
spanning tree would have been contracted after step~\eqref{enum:contract},
which leads to $|W|=1$ and contradicts our assumption $|W|>1$.

\begin{lemma}\label{lem:boundZ}
Let $\mathcal{L}$ be a maximum family of linearly independent
spanning tree constraints that are tight with respect to $x$.
Then
\begin{equation*}
2 z(F_2) \leq 2|\mathcal{L}| + 2(|W|-1)
-2(|F_0|+|F_1|) - 2x(F_0).
\end{equation*}
\end{lemma}
\begin{proof}
We can rewrite $2z(F_2)$ as follows by using the fact that
$x(F)=|W|-1$ (because $x$ is in the spanning tree polytope
of $H$).
\begin{equation}\label{eq:boundZ1}
\begin{aligned}
2z(F_2) &= 2z(F) - 2z(F_0) - 2z(F_1) \\
        &= 2(|F|-x(F)) - 2z(F_0) - 2z(F_1) \\
        &= 2|F| - 2(|W|-1) - 2z(F_0) - 2z(F_1)
\end{aligned}
\end{equation}
Using classical arguments we can bound the
size of the support of $x$, which is by assumption equal to $|F|$,
by the number of linearly independent tight constraints from the
spanning tree polytope and the degree polytopes $P_{N_w}$ for $w\in W$.
In particular $x$ is uniquely defined by the tight spanning tree
constraints $\mathcal{L}$ completed with some set $\mathcal{D}$
of linearly independent degree constraints, and
we have $|F|=|\mathcal{L}|+|\mathcal{D}|$.
The degree constraints $\mathcal{D}$ can be partitioned
into $\mathcal{D}_w$ for $w\in W$, where $\mathcal{D}_w$ are linearly
independent constraints of the matroid polytope $P_{N_w}$.
By Lemma~\ref{lem:indepDegConstraints}, $|\mathcal{D}_w|$ is
bounded by the sum of $x$ over all edges in $\delta(w)$ that are
contained in the degree constraint at $w$.
When summing these bounds up over all $w\in W$, each edge in
$F_2$ is counted exactly twice, and each edge
in $F_1$ exactly once. Hence,
\begin{align*}
|\mathcal{D}| \leq 2x(F_2) + x(F_1)= 2x(F)-x(F_1)-2x(F_0)
 =2(|W|-1)-x(F_1)-2x(F_0).
\end{align*}
Using $|F|=|\mathcal{L}|+|\mathcal{D}|$ and the above bound,
we obtain from~\eqref{eq:boundZ1}
\begin{align*}
2z(F_2)&\leq 2|\mathcal{L}| + 2(|W|-1) - 2\left( z(F_0)+
z(F_1)+2x(F_0)+x(F_1)\right)\\
&= 2|\mathcal{L}| + 2(|W|-1) - 2(|F_0|+|F_1|) - 2x(F_0),
\end{align*}
where the last inequality follows from $z(U)+x(U)=|U|$
for any $U\subseteq F$.
\end{proof}

The size of a family $\mathcal{L}$ of linearly independent
tight spanning tree constraints can easily be bounded by
$|W|-1$ using the fact that one can assume $\mathcal{L}$
to be laminar by standard uncrossing arguments (and
$\mathcal{L}$ contains no singleton sets).
However, this result shows not to be strong enough for
our purposes.
To strengthen this bound we exploit the fact that if
$\mathcal{L}$ contains close to $|W|-1$ sets, then
there are many nodes $w\in W$ that are ``sandwiched''
between two sets of $\mathcal{L}$, i.e., there are
two sets $L_1,L_2\in \mathcal{L}$ with $L_2=L_1\cup \{w\}$,
which in turn implies $x(\delta(w)\cap E[L_2])=1$.
Notice that for any degree two node $w$ which is not in $Q$,
we have $x(U)\neq 1$ for all $U\subseteq \delta(w)$.
Hence, such a node cannot be ``sandwiched'' between
two tight spanning tree constraints, and we expect
that the more such nodes we have, the smaller is $|\mathcal{L}|$.
The following result quantifies this observation.
It is stated in the general context of a spanning tree
polytope of a general connected graph
(not being linked to our degree-constrained problem).

\begin{lemma}\label{lem:stCard}
Let $y$ be a point in the spanning tree polytope for
a given graph $G=(V,E)$, and let
\begin{equation*}
S(G,y)=\{v\in V \mid |\delta(v)|=2,\ y(U)\neq 1 \;\forall U\subseteq \delta(v)\}.
\end{equation*}
Let $\mathcal{L}\subseteq 2^V$ be any linearly independent
family of
spanning tree constraints that are tight with respect to $y$.
Then
\begin{equation*}
|\mathcal{L}|\leq |V|-1-\left\lfloor \frac{1}{2}|S(G,y)|\right\rfloor.
\end{equation*}
\end{lemma}
\begin{proof}
To simplify notation let $S=S(G,y)$.
By standard uncrossing arguments
(see for example~\cite{goemans_2006_minimum}),
we can assume that $\mathcal{L}$ is laminar.
We first consider the case that there is a set
$L\in \mathcal{L}$ with $L\subseteq S$.
Let $L$ be a minimal set in $\mathcal{L}$ with this property.
Since $L$ is a tight spanning tree constraint, we have
that $y\big |_{E[L]}$ is in the spanning tree polytope of $G[L]$,
and hence $y(\delta(v)\cap E[L])\geq 1$ for $v\in L$.
As $L\subseteq S$, we have $|\delta(v)|=2$ and $y(e)<1$
for $v\in L$ and $e\in \delta(v)$. This implies that every
vertex in $L$ must have both of its neighbors in $L$
to satisfy $y(\delta(v)\cap E[L])\geq 1$. Since $G$ is
connected, as we assumed that there is a point in the spanning
tree polytope of $G$, we must have $L=V=S$.
Furthermore $|V|\geq 3$, because vertices in $L$ have degree two.
Hence the claim trivially follows since $|\mathcal{L}|=1$.

Now assume that there is no set $L\in \mathcal{L}$
with $L\subseteq S$.
We show that there exists a set $R\subseteq S$ of size
at least $|R|\geq \frac{1}{2}|S|$, such that the
laminar family $\mathcal{L}_R = \{L\setminus R \mid L \in \mathcal{L}\}$
over the elements $V\setminus R$ satisfies the following:
\begin{compactenum}[i)]
\item $\mathcal{L}_R$ has no singleton sets,
\item $|\mathcal{L}_R|=|\mathcal{L}|$, i.e.,
any two sets $L_1, L_2\in \mathcal{L}$ with $L_1\subsetneq L_2$,
satisfy $L_2\setminus L_1\not\subseteq R$.
\end{compactenum}
Notice that this will imply the claim since
$|\mathcal{L}_R|\leq |V\setminus R|-1$, because $\mathcal{L}_R$
is laminar without singleton sets, and hence
$|\mathcal{L}|=|\mathcal{L}_R|
\leq |V\setminus R|-1 \leq |V|-1-\frac{1}{2}|S|$.
It remains to define the set $R$ with the desired properties.
For $L\in \mathcal{L}$, let $V_L\subseteq L$ be all vertices
in $L$ that are not contained in any set $P\in\mathcal{L}$ with
$P\subseteq L$.
For each set $L\in \mathcal{L}$, include an arbitrary set
of $\lceil \frac{1}{2}|S\cap V_L| \rceil$ elements
of $S\cap V_L$ in $R$.
Since the sets $V_L$ for $L\in \mathcal{L}$ are a partition of
all vertices $V$, we clearly have $|R|\geq \frac{1}{2}|S|$.
Furthermore $R$ satisfies the desired properties as we show below.

i) Assume by sake of contradiction that $\mathcal{L}_R$ contains
a singleton set, i.e., there is a set $L\in\mathcal{L}$ with
$|L\setminus R|=1$. We can assume that $L$ is a minimal set
in $\mathcal{L}$.
By assumption we have $L\not\subseteq S$,
and since $R\subseteq S$, the element in $L\setminus R$ is not in $S$.
Hence, $R$ contains all elements $L\cap S$, which is only possible
if $|L\cap S|=1$ and therefore $|L|=2$. However, this implies that
there must be an edges of weight one between the two vertices in $L$,
which contradicts the fact that one of those vertices is in $S$.

ii) Assume by contradiction that there are two sets
$L_1, L_2\in \mathcal{L}$ with $L_1\subsetneq L_2$ that satisfy
$L_2\setminus L_1 \subseteq R$.
We can choose $L_1$ and $L_2$ such that there is no set $L\in \mathcal{L}$
with $L_1\subsetneq L \subsetneq L_2$. By choice of $R$, this can only happen if
$L_2\setminus L_1$ contains exactly one vertex $v\in S$.
This implies $y(\delta(v)\cap E[L_2])=1$, which contradicts the fact
that $v\in S$.
\end{proof}

Lemma~\ref{lem:stCard} can easily be generalized to the subgraph
of a given graph $G$ obtained by deleting the vertices $Q(G,y)$.
This form of the lemma is more useful for our analysis because
of our special treatment of vertices in $Q$.
\begin{lemma}\label{lem:stCardRed}
Let $y$ be a point in the spanning tree polytope of a given graph
$G=(V,E)$ with $y(e)\neq 1\; \forall\:e\in E$,
let $G'=G[V\setminus Q(G,y)]$, and let $y'$ be the projection
of $y$ to the edges in $G'$.
Let $\mathcal{L}$ be any linearly independent family
of spanning tree constraints of $G$ that are tight with respect to $y$. Then 
\begin{equation*}
|\mathcal{L}| \leq |V|-1 -
\left\lfloor \frac{1}{2}|S(G',y')|\right\rfloor.
\end{equation*}
\end{lemma}
\begin{proof}
By standard uncrossing arguments, we can assume that
$\mathcal{L}$ is a maximal laminar family of tight
spanning tree constraints.
We prove the result by induction on the number of elements in $Q=Q(G,y)$.
If $Q=\emptyset$, then the result follows from Lemma~\ref{lem:stCard}.
Let $q\in Q$ be a possible first element added to $Q$ during
the iterative construction of $Q$, i.e., $y(\delta(q))=1$.
This implies that $V\setminus \{q\}$ is a tight spanning tree constraint.
Let $H=G[V\setminus \{q\}]$, $y_H=y\big\vert_{E[V\setminus\{q\}]}$
and $Q_H=Q(H,y_H)$.
Since $Q_H=Q\setminus \{q\}$, we can apply
the induction hypothesis to the graph $H$
to obtain that any maximal family $\mathcal{L}_H$ of
linearly independent tight spanning tree constraints in $H$
with respect to $y_H$ satisfies
$|\mathcal{L}_H|\leq |V\setminus \{q\}|-1-
\left\lfloor\frac{1}{2}|S(G',y')|\right\rfloor$.
The claim follows by observing that
$\mathcal{L}=\mathcal{L}_H\cup \{V\}$ is a maximal family
of tight spanning tree constraints in $G$, and hence
\begin{equation*}
|\mathcal{L}|=|\mathcal{L}_H|+1 \leq |V|-1
-\left\lfloor\frac{1}{2}|S(G',y')|\right\rfloor.
\end{equation*}
\end{proof}

Combining Lemma~\ref{lem:stCardRed} with Lemma~\ref{lem:boundZ} 
we obtain the following bound, where we use 
$S=S(H[W\setminus Q],x\big\vert_{F[W\setminus Q]})$ to simplify
the notation.
To get rid of the rounding on $\frac{1}{2}|S|$ we use
$2 \lfloor \frac{1}{2} |S|\rfloor\geq |S|-1$.

\begin{corollary}\label{cor:tokenCount}
\begin{equation*}
2z(F_2)\leq 4(|W|-1) - 2(|F_0|+|F_1|) - 2x(F_0) - |S|+1.
\end{equation*}
\end{corollary}

The following lemma implies the correctness of our algorithm.
We recall that $Y\subseteq W$ is the set of all nodes $w\in W$
such that $\delta(w)\cap F_2\neq \emptyset$.

\begin{lemma}
There is a node $w\in Y$ such that a type A constraint adaptation can
be applied to $w$.
\end{lemma}
\begin{proof}
Let $\overline{Y}=W\setminus Y$. We will prove that 
\begin{equation}\label{eq:toProveTokenCount}
4|\overline{Y}| \leq  2(|F_0|+|F_1|) + 2x(F_0) + |S|.
\end{equation}
Together with Corollary~\ref{cor:tokenCount} this then implies
$2z(F_2)\leq 4|Y|-3$, which in turn implies by an averaging argument that
there is at least one node in $Y$ to which a type A constraint adaptation
can be applied.
To prove~\eqref{eq:toProveTokenCount} we apply a fractional
token counting argument: we show that if we interpret the
right-hand side of~\eqref{eq:toProveTokenCount}
as a (fractional) amount of tokens, then we can assign those tokens
to the vertices in $\overline{Y}$ such that each vertex in
$\overline{Y}$ gets at least $4$ tokens.

We think of the tokens corresponding
to $2(|F_0|+|F_1|)+2x(F_0)$ as residing
at the endpoints of the edges in $F_0\cup F_1$. Each edge $f\in F_0$
gets $2+2x(f)$ tokens, $1+x(f)$ at each endpoint.
Each edge $f\in F_1$ gets $1+x(f)$ tokens at the endpoint which does
not contain $f$ in its degree constraint, and $1-x(f)$ tokens at
the other endpoint.
The tokens assigned to the endpoints of the edges thus sum up
to $2(|F_0|+|F_1|)+2x(F_0)$.

We start by assigning tokens to vertices in $Q$. By definition of the
vertices in $Q$, we can order the elements in $Q=\{q_1,\dots, q_p\}$ such
that for $i\in [p]$, we have $x(F_{q_i})=1$ where
$F_{q_i}=\{\{q_i,v\}\in F \mid v\in W\setminus \{q_1,\dots, q_{i-1}\}\}$.
Since $x(F_{q_i})=1$ and no edge $f\in F$ satisfies $x(f)=1$ (such an
edge would have been contracted), we have $|F_{q_i}|\geq 2$.
Each vertex $q_i\in Q$ gets all the tokens at both endpoints
of the edges in $F_{q_i}$. Since $|F_{q_i}|\geq 2$, $q_i$ receives
indeed at least four tokens.

Let $H'=(W',F')=H[W\setminus Q]$ be the induced subgraph over
the vertices $W\setminus Q$
, and let $x'=x\big\vert_{F'}$.
Notice that $x'$ is in the spanning tree polytope of $H'$
since the set of edges $U\subseteq F$ that have at least one endpoint
in $Q$ satisfy $x(U)=|U|$, and hence $x'(F')=|F'|-1$.

The remaining tokens are allocated as follows. Each node
$w\in \overline{Y}\cap W'$ gets for every edge $f\in \delta_{H'}(w)$,
the tokens of $f$ at the endpoint at $w$.
Furthermore, every node in $S$ gets an additional token
from the term $|S|$.

The attributed tokens clearly do not exceed
the right-hand side of~\eqref{eq:toProveTokenCount}.
It remains to show that each node $w\in \overline{Y}\cap W'$
gets at least $4$ tokens.
We distinguish the following three cases: (i) $w\in S$, (ii) $w\not\in S$
and none of the edges $\delta_{H'}(w)$ is contained in the degree
constraint at $v$,
and (iii) $w\not\in S$ and at least one edge of $\delta_{H'}(w)$
is contained in the degree constraint at $w$.
Notice that the vertices considered in case (i) are precisely all
vertices in $H'$ of degree two,
because if there was a degree two vertex $w\in W'\setminus S$,
then $w$ would have been included in $Q$.
Hence, all vertices considered in case (ii) or case (iii) have degree
at least $3$ in $H'$.

\emph{Case (i): $w\in S$.}
Because $|\delta_{H'}(w)|= 2$, we have that both edges
in $\delta_{H'}(w)$ are not contained in the degree constraint at $w$,
since otherwise a type B degree adaptation could have
been performed at $w$. Hence, $w$ receives $2+x(f_1)+x(f_2)$ tokens from
those two edges plus one token from $|S|$, resulting in $3+x(f_1)+x(f_2)$
tokens. Since $x'$ is in the spanning tree polytope
of $H'$, we have $x(f_1)+x(f_2)=x(\delta_{H'}(w))\geq 1$,
and thus $w$ receives at least $4$ tokens.

\emph{Case (ii): $w\not\in S$ and none of the edges $\delta_{H'}(w)$
is contained in the degree constraint at $w$.}
The total number of tokens received by $w$ thus equals
$|\delta_{H'}(w)|+x(\delta_{H'}(w))\geq 3 + x(\delta_{H'}(w))$,
since $|\delta_{H'}(w)|\geq 3$.
The claim follows again by observing that $x'$
is in the spanning tree polytope of $H'$, which implies
$x(\delta_{H'}(w))\geq 1$.

\emph{Case (iii): $w\not\in S$ and at least one edge of
$\delta_{H'}(w)$ is contained in the degree constraint at $w$.}
Let $U$ be the set of all edges in $\delta_{H'}(w)$
that are contained in the degree constraint at $w$.
Since no type B degree adaptation can be performed
at $w$, we have $|U|-x(U)>4$.
However, $|U|-x(U)$ is exactly the number of tokens that
$w$ receives from the edges in $U$. Hence, at least $4$
tokens are assigned to $w$.
\end{proof}

\bibliographystyle{plain}
\bibliography{literature}

\end{document}